\documentclass[letterpaper, 10 pt, conference]{ieeeconf}   
\IEEEoverridecommandlockouts                      
\usepackage{amsmath}
\usepackage[english]{babel} 
\usepackage{amssymb}
\usepackage{dsfont}
\usepackage{graphicx}
\usepackage{theorem}
\usepackage{cite}
\usepackage{pstricks}
\usepackage{booktabs}
\usepackage{epsfig}
\usepackage{pst-grad} 
\usepackage[normalem]{ulem} 
\usepackage{bbm}
\usepackage{float}
\usepackage{footnote}
\usepackage{siunitx}
\DeclareMathOperator{\im}{im}
\DeclareMathOperator{\col}{col}
\DeclareMathOperator{\diag}{diag}

{\theorembodyfont{\itshape}\newtheorem{theorem}{Theorem}}
{\theorembodyfont{\itshape}}
{\theorembodyfont{\itshape}\newtheorem{lemma}{Lemma}}
{\theorembodyfont{\itshape}}
{\theorembodyfont{\upshape}}
{\theorembodyfont{\itshape}}
{\theorembodyfont{\upshape}\newtheorem{remark}{Remark}}
{\theorembodyfont{\upshape}}
{\theorembodyfont{\upshape}}
{\theorembodyfont{\upshape}\newtheorem{assumption}{Assumption}}

\newcommand{\1}{\mathds{1}}
\newcommand{\R}{\mathbb{R}}
\newcommand{\w}{\omega}

\newcommand{\half}{\frac{1}{2}}
\newcommand{\la}{\mathcal{L}}

\newcommand{\ac}{\rm ac}
\newcommand{\dc}{\rm dc}

\newcommand{\G}{\mathcal{G} }
\newcommand{\V}{\mathcal{V} }
\newcommand{\E}{\mathcal{E} }
\newcommand{\N}{\mathcal{N} }

\usepackage[prependcaption,colorinlistoftodos]{todonotes}

\usepackage{tikz}
\usepackage{tkz-berge}
\usetikzlibrary{positioning}
\usetikzlibrary{arrows,%
	petri,%
	topaths}%
\usetikzlibrary{decorations.markings}         
\tikzstyle{vertex}=[circle, shading = ball, ball color = white!100!white, minimum size = 15pt, draw, inner sep=0pt]  
\newcommand{\vertex}{\node[vertex]}           
\newcommand{\weight}[1]{{\footnotesize $\mathit{#1}$}}

\tikzset{
	LabelStyle/.style = { rectangle, rounded corners, draw,
		minimum width = 2em, fill = yellow!50,
		text = red, font = \bfseries },
	VertexStyle/.append style = { inner sep=5pt,
		font = \Large\bfseries},
	EdgeStyle/.append style = {->, bend left} }
\usepackage[normalem]{ulem} 

\newcommand{\rem}[1]{{\color{red}\sout{#1}}}     

\renewcommand{\rem}[1]{}                      

\begin{document}
	\title{\Large \bf Stability and Frequency Regulation of Inverters with Capacitive Inertia}
	
	\author{Pooya~Monshizadeh,\; Claudio De Persis,\; Tjerk Stegink,\; Nima Monshizadeh,\; and Arjan van der Schaft%
		\thanks{Pooya Monshizadeh and Arjan van der Schaft are with the Johann Bernoulli Institute for Mathematics and Computer Science, University of Groningen, 9700 AK, the Netherlands,
			{\tt\small p.monshizadeh@rug.nl, a.j.van.der.schaft@rug.nl}}%
		\thanks{Claudio De Persis and Tjerk Stegink are with the Electronics, Power and Energy Conversion Group, University of Groningen, 9747 AG, the Netherlands,
			{\tt\small  c.de.persis@rug.nl, t.w.stegink@rug.nl}}%
		\thanks{Nima Monshizadeh is with the Electrical Engineering Division, University of Cambridge, CB3 0FA, United Kingdom,
			{\tt\small n.monshizadeh@eng.cam.ac.uk}}%
		\thanks{This work is supported by the STW Perspectief program "Robust Design of Cyber-physical Systems" under the auspices of the project "Energy Autonomous Smart Microgrids".}
	}
	\maketitle
	\begin{abstract}
		In this paper, we address the problem of stability and frequency regulation of a recently proposed inverter. In this type of inverter, the DC-side capacitor emulates the inertia of a synchronous generator. First, we remodel the dynamics from the electrical power perspective. Second, using this model, we show that the system is stable if connected to a constant power load, and the frequency can be regulated by a suitable choice of the controller. Next, and as the main focus of this paper, we analyze the stability of a network of these inverters, and show that frequency regulation can be achieved by using an appropriate controller design. Finally, a numerical example is provided which illustrates the effectiveness of the method.
	\end{abstract}
\section{Introduction}
Along with the emergence of the renewable energy sources in power networks, and consequently the increasing usage of power converters, new issues and concerns regarding stability of the grid have arisen. Recently, the problem of low inertia of inverter dominated systems has been extensively investigated. In classical electrical grids, synchronous generators dominated the power source types in the network. These machines possess a massive rotational part, rotating at the same frequency as that of the generated electrical sinusoidal voltage. The kinetic energy of such rotation takes the role of an energy reservoir. When an abrupt increase or decrease occurs in the load, the kinetic energy of the synchronous machine is injected into, or absorbed from the network, respectively.
In conventional power converters, the absence of this reservoir jeopardizes the stability of the network, and leads to new frequency instability issues in power systems \cite{Tielens2016,Ulbig2014,Bevrani2014}. Inverters possess fast frequency dynamics and the traditional control strategies are too slow to prevent large frequency deviations and their consequences \cite{Ulbig2014}. In particular, in networks with low inertia, the rate of change of frequency (ROCOF) may be large enough to activate the load-shedding switches of a power network, even with a small power imbalance \cite{Dreidy2017}. As a remedy to this problem, the concept of \textit{Virtual Inertia} has been introduced and various methods have been proposed so that the inverters emulate the behavior of synchronous generators \cite{vi-4,vi-3,vi-2,vi-1,vi0,vi1,vi2,vi3,vi4}. 
\begin{figure}
	\centering
	\includegraphics[width=8.5cm]{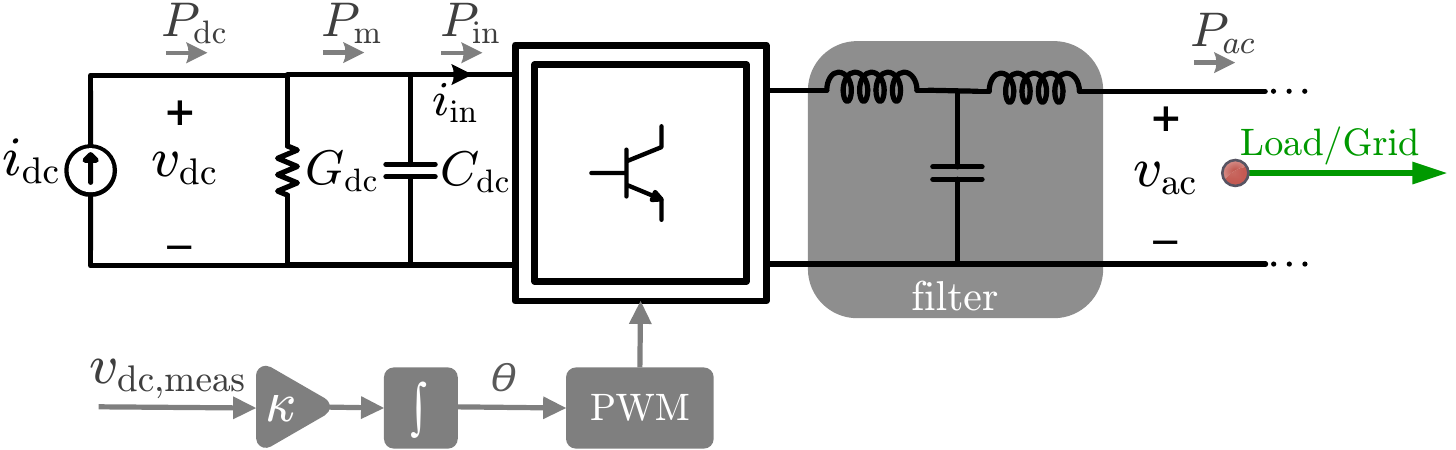}
	\caption{Schematic of an inverter with capacitive inertia (ICI)}
	\label{fig:circuit}
	\vspace{-0.3cm}
\end{figure}

Although a better performance of the inverters results with this emulation, the virtual inertia cannot react instantaneously. This is due to the fact that the AC measurements play a major role in mimicking the inertia \cite{Dreidy2017}, and hence the inevitable delay in these measurements slows down the emulating behavior. Therefore, as an alternative, methods to provide an instantaneous physical inertia have been proposed. More specifically in \cite{Vandoorn2011,Vandoorn2012,Arani2013}, the energy stored in the DC-side capacitor of the inverter is employed as a replacement of the kinetic energy stored in the rotor of a synchronous generator. The DC-side capacitor is an inherent element in most inverters. We refer to these devices as \textit{Inverters with Capacitive Inertia} (ICI) throughout the paper. Recently, a promising and detailed nonlinear model of such devices is provided in \cite{Jouini2016}, where the generated frequency is proposed to be proportional to the measured voltage of the DC-side capacitor (see Figure \ref{fig:circuit}). However in \cite{Jouini2016}, the stability of the inverter, connected to a single load or a network, was not investigated. Note that, as previously mentioned, the motivation for emulating inertia is to alleviate the stability problems of low-inertia networks dominated by inverters. 

In this paper, we remodel the ICI dynamics in \cite{Jouini2016} from a power perspective, in order to ease the stability analysis of these devices in several scenarios. In Section \ref{s:SingleICI}, the case of a single inverter connected to a constant power load will be investigated. A primary controller is provided, which guarantees stability of the system. Next, it is shown that the frequency can be regulated to its nominal value by a secondary controller. In Section \ref{s:Network}, stability of a network of ICIs is investigated and a distributed controller is proposed to regulate the frequencies to the desired value. Finally, a numerical example illustrates the effectiveness of the method.

\textbf{Notation  }For $i \in \{1, 2, . . . , n\}$, by $\col(a_i)$ we denote the column
vector $[a_1\; a_2 \; \cdots \; a_n]^T$. For a given vector $a \in \R^n$, the diagonal matrix $\diag\{a_1,a_2,\cdots , a_n\}$ is
denoted in short by $[a]$. The function $\sin a$ represents the element-wise sine function, i.e. $\sin a=\col(\sin (a_i))$. The symbol $\1$ denotes the vector of ones with an appropriate dimension, and $I_n$ is the identity matrix of size $n$.
\section{Single Inverter with Capacitive Inertia}\label{s:SingleICI}
In this section, we first explain briefly how a single inverter is modeled in \cite{Jouini2016}, and next we reconfigure the model from the electrical power perspective. Finally the control method is elaborated.
\subsection{ICI Model in \cite{Jouini2016}}
Figure \ref{fig:circuit} depicts the schematic of an ICI, which is based on the averaged model of a three-phase converter (For the sake of clarity, the electrical circuit of one phase is shown.). 
The electrical part, shown in black, consists of a controllable current source $i_{\dc}$, a resistor with the conductance $G_{\dc}$, and a capacitor $C_{\dc}$ in the DC-side. The switching block in the middle converts the DC current to an alternating current. This conversion is carried out via a pulse width modulation (PWM) unit which provides on/off signals to the switching block according to a given phase angle input $\theta$. A low-pass $LCL$ filter in the AC-side eliminates the high frequency harmonics of the output signal. This process generates a sinusoidal voltage $v_{\ac}$ with the phase angle $\theta$. 

In a synchronous generator, when the power demand is more than the mechanical input power, the lacking amount of energy is taken from the kinetic energy of the rotor $(\half J\w^2)$, hence the angular velocity of the rotor decreases and the frequency of the output voltage drops. 
Similarly, in power converters with a DC-side capacitor, the extra power demand is released from the energy $\half C_{\dc}v^2_{\dc}$ stored in the capacitor. However, contrary to the inertia of a synchronous generator, if the voltage of the DC-side drops, this will not be visible in the output frequency at the AC-side. In order to remedy this, in \cite{Jouini2016}, to emulate the inertial behavior, the frequency $\w=\dot{\theta}$ of the output voltage $v_{\ac}$ is designed to be proportional to $v_{\dc}$. This is achieved via an integral action over the measured voltage $v_{\dc}$ with the integral coefficient $\kappa$, and feeding it as the PWM signal to the switching block, i.e. $\dot{\theta}=\kappa v_{\dc}$ (see Figure \ref{fig:circuit}). Hence
\begin{align}\label{e:omega}
\w&=\kappa v_{\dc}\;,
\end{align}
where a reasonable choice for the integral coefficient is $\kappa= \frac{\w^*}{v_{\dc}^*}$, with $\w^*\in\R$ denoting the desired frequency (angular velocity corresponding to $\SI{50}{\hertz}$ or $\SI{60}{\hertz}$). Furthermore, using Kirchhoff's current law in the DC side, we have
\begin{align}\label{e:KCL}
C_{\dc}\dot{v}_{\dc}=-G_{\dc} v_{\dc}-i_{\rm in}+i_{\dc}\;.
\end{align}
Combining \eqref{e:omega} and \eqref{e:KCL} we obtain the model \cite{Jouini2016}
\begin{align}\label{e:pre-model}
J\dot{\w}=-D \w-\frac{i_{\rm in}}{\kappa}+\frac{i_{\dc}}{\kappa}\;,
\end{align}
where $J=\frac{C_{\dc}}{\kappa^2}$ and $D=\frac{G_{\dc}}{\kappa^2}$. 
\medskip
\subsection{ICI Model from the Electrical Power Perspective}
We can rewrite the system \eqref{e:pre-model} as
\begin{align*}
J\dot{\w}&=-D\w-\frac{P_{in}}{\kappa v_{\dc}}+\frac{i_{\dc}}{\kappa}\;,
\end{align*}
where $P_{in}=v_{\dc}i_{\rm in}$ is the electrical power that is injected into the switching block. 
Assuming that no power is dissipated in the switching block and the $LCL$ filter (see Figure \ref{fig:circuit}), i.e. $P_{in}\simeq P_{\ac}$, we obtain
\begin{align}\label{e:model}
J\dot{\w}&=-D\w-\frac{P_{\ac}}{\w}+u\;,
\end{align}
where $u=\kappa^{-1}i_{\dc}$ is treated as the control input. 
\subsection{Primary Control}\label{ss:PC}
Consider an ICI modeled by \eqref{e:model} connected to a constant
power load $P_{\ac}=P_\ell$. To provide a primary control, we
propose the control input 
\begin{align}\label{e:uS}
u=D\w^*+\w^{-1}P_{m}\;,
\end{align}
where $P_m \in \R_{>0}$ will be designed later. This design is inspired by the following remark.
\begin{remark}
	Around the
	nominal frequency $\w=\w^*$, the term $P_m$ in \eqref{e:uS} represents the power injection behind the capacitor $C_{\dc}$ (see Figure 	\ref{fig:circuit}). To see this, notice that we can rewrite \eqref{e:uS} as
	$$ \frac{v^*_{\dc}}{\w^*}i_{\dc}=G_{\dc}\frac{v_{\dc}^{*2}}{{\w^*}}+\frac{P_m}{\w^*}
	\;,$$
	where we used $u=\kappa^{-1}i_{\dc}$, $D=\kappa^{-2}G_{\dc}$, and $\kappa=\frac{\w^*}{v^*_{\dc}}$. Hence we have
	$$
	P_m=P^*_{\dc}-G_{\dc}v^{*2}_{\dc}\;.
	$$
	Note that the first term is the nominal DC power, and the second term is the power dissipated in the DC-side resistor in the nominal frequency.
\end{remark}
Since $\w=\kappa v_{\dc}$, where $v_{\dc}$ is a DC value measured for generating the PWM signal, no additional measurement is required to implement this controller.
In this section, we assume a constant
$P_m={P_\ell}^*$, where ${P_\ell}^*>0$ is an estimate of the nominal load.
Now, the model
\eqref{e:model} can be rewritten as
\begin{align}\label{e:model0}
J\dot{\w}&=-D(\w-\w^*)+\frac{{P_\ell}^*-P_{\ell}}{\w}\;.
\end{align}
The model \eqref{e:model0} indicates a droop-like behavior. That is, the frequency will drop if the power extracted by the load is larger than the nominal power, and will increase otherwise. In fact, the dynamics \eqref{e:model0} resembles that of a synchronous generator modeled with an \textit{improved swing equation} \cite{Zhou2009},\cite{Pooya2016}, with inertia $J=\frac{C_{\dc}}{\kappa^2}$, damping coefficient $D=\frac{G_{\dc}}{\kappa^2}$, and mechanical input power ${P_\ell}^*$. Assume that the maximum power mismatch (lack of power) $P_\ell-{P_\ell}^*$ is such that
 \begin{equation*}
 \Delta:=\w^{*2}-4\frac{P_\ell-{P_\ell}^*}{D}>0\;.
 \end{equation*}
Then the dynamics \eqref{e:model0} has the following two equilibria
 \begin{align}\label{equib}
 & {\w}_s=\half(\w^*+\sqrt{\Delta}) ,\quad
 {\w}_u=\half(\w^*-\sqrt{\Delta}) 
 \; \text{.} 
 \end{align}
The system is stable around the equilibrium point $\w= \w_s$ (see Theorem 1 in \cite{Pooya2016} for a proof and more details). A secondary controller is needed to eliminate the static deviation of $ \w_s$ from the nominal frequency $\w^*$.
\begin{remark}
Aiming at a larger damping coefficient $(D)$ requires a larger $G_{\dc}$ and consequently more power loss $(v_{\dc}^2G_{\dc})$ in the DC-side resistor. Therefore, in the case that a larger damping term $D(\w-\w^*)$ in \eqref{e:model0} is desired, a proportional controller term can be added to the control input. In particular, let $u=D\w^*+\w^{-1}P_m+u_p$, where $u_p=\tilde{D}(\w-\w^*)$ for some $\tilde{D}>0$. In this case, the damping term in \eqref{e:model0} modifies to $(D+\tilde{D})(\w-\w^*)$. 
\end{remark}
\subsection{Secondary Control}
Aiming at the frequency regulation of the system \eqref{e:model}, we propose the controller as
\begin{equation}
\begin{aligned}
\dot \chi&=-\w^{-1}(\w-\w^*)\\
u&=D \w^*+\w^{-1}\chi  \;.
\end{aligned}\label{e:model0c}
\end{equation}
Note that here, compared to the primary controller, the term $P_m$ in \eqref{e:uS} is not a constant, but a state variable integrating the frequency deviation. This controller regulates the frequency to the nominal $\w^*$ in the steady state of the system \eqref{e:model}  (see Remark \ref{r:Network2Single} later on). 

\section{Network of Inverters with Capacitive Inertia}\label{s:Network}
In this section, we investigate the stability and the frequency regulation in a network of ICIs. 
\subsection{Model}
Consider an inverter-based network, where each bus is connected to an inverter and a local constant power load $P_\ell$. 
The topology of the grid is represented by
a connected undirected graph $\G(\V, \E)$, with node set $\V$, and edge set $\E$, given by a set of unordered pairs $\{i, j\}$ of distinct vertices $i$ and $j$. Let $n=|\V|$ and $m=|\E|$. By
assigning an arbitrary orientation to the edges, the incidence matrix $B\in\R^{m\times n}$ is defined element-wise as
$B_{i\ell} = 1$, if node $i$ is the sink of the $\ell$th edge, $B_{i\ell} =-1$, if $i$ is the source of the $\ell$th edge and $B_{i\ell} = 0$ otherwise.
Due to the inductive output impedance of the inverters, the lines are assumed to be dominantly inductive \cite{Schiffer2014, Pooya2017}, i.e. two nodes $\{i,j\}\in\E$ are connected by a nonzero inductance. The set of neighbors of the $i$th node is denoted by $\N_i =\{j \in \V \ |\ \{i, j\}\in \E\}$. 

Calculation of the active power transferred via a power line is in general cumbersome, and complicates the network stability analysis. To remove this obstacle, we take advantage of phasor approximations. 
The relative phase angles are
denoted in short by $\theta_{ij} := \theta_i-\theta_j \;, \{i, j\} \in \E$. 
Now let $\gamma_k:=\frac{|V_i||V_j|}{X_{ij}}\,,k \sim \{i,j\}$, where $X_{ij}$ represents the reactance of the line connecting nodes $i$ and $j$, and $|V_i|$ denotes the magnitude of the voltage at node $i$ and is assumed to be constant. Then the active power transferred via the inductor between nodes $i$ and $j$ is calculated as $$P_{ij}=\gamma_k \sin \theta_{ij}\;\;,\;k \sim \{i,j\}\;.$$ 
Hence, the injected active power by the inverter at each node $P_{{\ac}_i}$ is given by
\begin{align}\label{e:power}
P_{{\ac}_i} = P_{\ell_i}+\sum_{\substack{j\in\N_i\\ k \sim \{i,j\}}} \gamma_k \sin \theta_{ij}\;\;,
\end{align}
where $P_{\ell_i}$ denotes the local load connected to node $i$. Note that the phasor approximation is only exploited to write the expression of the active power above. 

For every node $i$ we have
\begin{equation*}
\begin{aligned}
\dot \theta_i&=\w_i\\
J_i  \dot \w_i&=u_i-\w_i^{-1}P_{\ac_i}-D_i\w_i\;.
\end{aligned}
\end{equation*}
With a little abuse of notation, using \eqref{e:power}, the network can be written in vector form as
\begin{equation}\label{e:delta}
\begin{aligned}
\dot \theta&=\w\\
J  \dot \w&=u-[\w]^{-1}(P_\ell+B\Gamma\sin(B^T\theta))-D\w \;,
\end{aligned}
\end{equation}
where $\theta=\col(\theta_i)$, 
$\w=\col(\w_i)$,
$J=\diag\{J_1,\cdots,J_n\}$, $u=\col(u_i)$,
$P_\ell=\col(P_{\ell_i})$, $D=\diag\{D_1,\cdots,D_n\}$, and $\Gamma=\diag\{\gamma_1,\cdots,\gamma_m\}$,  with indices indicating the node/edge numbers.

Note that if $(\theta, \omega)$ is a solution to \eqref{e:delta} for given $u$ and $P_\ell$, then $(\theta+\1 \alpha, \omega)$ is also a solution to \eqref{e:delta} for any constant $\alpha\in \R$. To exclude this rotational invariance, it is convenient to introduce a different set of coordinates, representing the phase angle differences, given by $\eta:=B^T\theta$. Then the model \eqref{e:delta} modifies to
\begin{equation}\label{e:network}
\begin{aligned}
\dot \eta&=B^T\w\\
J  \dot \w&=u-[\w]^{-1}(P_\ell+B\Gamma\sin \eta)-D\w\;.
\end{aligned}
\end{equation}

\subsection{Primary Control}
The goal of primary control is to design a proportional controller $u=k(\omega)$ such that frequency variables converge to the same value corresponding to a stable equilibrium of the system. To this end, analogous to the case of a single ICI, and with a little abuse of notation, we propose the control input 
\begin{align}\label{e:u}
u=D\1\w^*+[\w]^{-1} P_m\;.
\end{align}
For a constant setpoint $P_m=P^*_\ell=\col(P^*_{\ell_i})$, the dynamics \eqref{e:network}
reads as
\begin{equation}\label{e:pnetwork}
\begin{aligned}
\dot \eta&=B^T\w\\
J  \dot \w&=[\w]^{-1}({P_\ell}^*-P_\ell-B\Gamma\sin \eta)-D(\w-\1\w^*)\;.
\end{aligned}
\end{equation}
Note that $\eta(0)=B^T\theta(0)$, and hence $\eta(t)\in\im B^T$ for all $t\geq0$.
Hence, we can restrict the domain of solutions to $(\eta,\w)\in\mathcal X:=\im B^T\times\R^n$, which is clearly forward invariant. 

Note that the choice of the setpoint ${P_\ell}^*$ is decided based on an estimate of the load $P_\ell$. We assume that the maximum mismatch (lack of power) $P_\ell-{P_\ell}^*$ is such that
\begin{equation}\label{e:Delta-multi}
\Delta_N:=\w^{*2}-\frac{4\1^T(P_\ell-{P_\ell}^*)}{\1^TD\1}>0.
\end{equation}
It is easy to see that the condition \eqref{e:Delta-multi} is necessary for the existence of an equilibrium for system \eqref{e:pnetwork}.
Next, we characterize the equilibria of \eqref{e:pnetwork}. 
\begin{lemma}\label{l:existS}
Assume that \eqref{e:Delta-multi} holds. Then the points $(\eta_s, \1\w_s)$ and $(\eta_u, \1\w_u)$ are two equilibria of system \eqref{e:pnetwork} if and only if
 \begin{align}\label{e:barS}
{P_\ell}^*-P_\ell&=B\Gamma \sin \eta_s +D\1\w_{s}(\w_{s}-\w^*),\\
\nonumber
{P_\ell}^*-P_\ell&=B\Gamma \sin\eta_u +D\1\w_{u}(\w_{u}-\w^*),
\end{align}
where
 \begin{align}\label{e:roots}
 & {\w}_s=\half(\w^*+\sqrt{\Delta_N}) ,\quad
 {\w}_u=\half(\w^*-\sqrt{\Delta_N}) 
 \; \text{,} 
 \end{align}
\end{lemma}
\begin{proof}
	By the first equality in \eqref{e:pnetwork} it follows that $\w=\1\tilde{\w}$ for some $\tilde{\w}$. By premultiplying the second equality in \eqref{e:pnetwork} by $\1^T$ we obtain that $\1^T({P_\ell}^*-P_\ell)=\1^TD\1\tilde{\w}(\tilde{\w}-\w^*)$ which is a quadratic equation with the roots given by \eqref{e:roots}. 
\end{proof}
The equilibrium of interest here is $(\eta_s, \1\w_s$). In fact, the other equilibrium can be shown to be unstable. Lemma \ref{l:existS} imposes the following assumption:
\begin{assumption}\label{assump:keyS}
For given $P_\ell$ and $P_\ell^\ast$, the inequality \eqref{e:Delta-multi} holds, and there exists $\eta_s\in  \im B^T \cap (-\frac{\pi}{2}, \frac{\pi}{2})^m$ such that \eqref{e:barS} is satisfied.
\end{assumption}
The additional constraint $\eta_s \in (-\frac{\pi}{2}, \frac{\pi}{2})^m$ is needed for stability of the equilibrium and is ubiquitous in the literature, often referred to as the \textit{security constraint} \cite{Hierarchy2016}.
To prove stability of the equilibrium $(\eta_s,\1\w_s)$, we consider first the energy function 
\begin{align}\label{e:VS}
V(x)=\frac12\w^TJ\w-{\w_s}^{-1}\1^T\Gamma \cos\eta\;,
\end{align}
with $x=\col(\eta,\w)$. 
Inspired by \cite{Bregman1967, Bayu2007, Nima2015, Trip2016}, we shift this energy function to
\begin{align}\label{e:WS}
V_s(x)=V(x)-(x-\bar x)^T\nabla V(\bar x)-V(\bar x)\;.
\end{align}
where $\bar x=(\eta_s,\1\w_s)$ and $\nabla V(\bar x)$ is the gradient of $V$ with respect to $x$ evaluated at $\bar x$.
By construction, $V_s$ is positive definite locally if the function $V$ is strictly convex around $\bar x$ \cite{Bregman1967}. By calculating the first and second partial derivatives of $V_s$, it is easy to observe that $V_s$ is strictly convex and takes its minimum at $x=\bar x$, provided that $\eta_s \in  (-\frac{\pi}{2}, \frac{\pi}{2})^m$. 
Now, we are ready to state the main result of this subsection:

\begin{theorem}\label{thm:localasstabS}
	Suppose that Assumption \ref{assump:keyS} holds. 
	Then there exists a neighborhood $\Omega$ of $(\eta_s, \1\w_s)$ such that any solution $(\eta, \w)$ to \eqref{e:pnetwork} that starts in $\Omega$, asymptotically converges to the equilibrium point $(\eta_s, \1\w_s)$.
\end{theorem}
\begin{proof}
	First, observe that by substituting ${P_\ell}^*-P_{\ell}$ from \eqref{e:barS} the system \eqref{e:pnetwork} can be written as 
	\begin{equation}\label{e:networkcshiftedS}
	\begin{aligned}
	\dot \eta=&B^T(\w-\1\w_s)\\
	J  \dot \w=&-[\w]^{-1}\Big(B\Gamma(\sin \eta-\sin\eta_s))
	\\&+D\big([\w](\w-\1\w^*)-\1\w_s(\w_s- \w^*)\big) \Big)\;.
	\end{aligned}
	\end{equation}
	Consider the Lyapunov function $V_s$ given by \eqref{e:VS}-\eqref{e:WS}. Computing the time derivative of $V_s$ along the solutions of \eqref{e:networkcshiftedS} yields
	\begin{align*}
	\dot {V_s}=&-{(\w-\1\w_s)^T[\w]^{-1}}\Big(B\Gamma(\sin \eta-\sin\eta_s))
	\\&\hspace{2.2cm} +D\big([\w](\w-\1\w^*)-\1\w_s(\w_s- \w^*)\big) \Big)\\
	&+{\w_s}^{-1}\big(\Gamma(\sin \eta-\sin \eta_s)\big)^TB^T(\w-\1\w_s)\\
	=&-(\w-\1\w_s)^T\big({[\w]}^{-1}B\Gamma(\sin\eta-\sin\eta_s)\\& \hspace{2.2cm}-{\w_s}^{-1}B\Gamma(\sin\eta-\sin\eta_s)\big)\\ 
	&-(\w-\1\w_s)^T[\w]^{-1}D[\w+\1\w_s-\1\w^*](\w-\1 \w_s)\;.
	\end{align*}
	Bearing in mind that $\w^*-\w_s=\w_u$, where $\w_u$ is given by \eqref{e:roots}, we have
	\begin{align*}
	\dot V_s=&(\w-\1\w_s)^T([\w]-{\w_s I_n})[\w]^{-1}{\w^*}^{-1}B\Gamma(\sin\eta-\sin\eta_s)\\
	&-(\w-\1\w_s)^TD[\w]^{-1}([\w-\1\w_u])(\w-\1 \w_s)\;.
      \end{align*}
	Hence, we obtain
	\begin{align*}
	\dot{ V_s}=&-(\w-\1\w_s)^T[\w]^{-1}\\
	&\Big(D[\w-\1\w_u]-{\w_s}^{-1}[z(\eta)]\Big)(\w-\1\w_s)
      \end{align*}
	with
	\[
	z(\eta):=B\Gamma(\sin\eta-\sin\eta_s)\;.
	\]
	Since $D>0$, $\w_s>0$, $[\1\w_s-\1\w_u]=\sqrt{\Delta_N}I_n>0$, and $z(\eta_s)=0$, there exists a neighborhood $\Omega^+$ around $(\eta_s,\1\w_s)$ such that 
	\begin{align*}
	[\w-\1\w_u]>0\;,\;\;\;D[\w-\1\w_u]-{\w_s}^{-1} [z(\eta)]>0
	\end{align*}
	for all $(\eta,\w)\in \Omega^+$. 
	Take a (nontrivial) compact level set $\Omega$ of $V_s$ contained in this set, i.e. $\Omega\subset \Omega^+$. Note that such $\Omega$ always exists for sufficiently small $r>0$, $\Omega=\{x \mid x\in \Omega^+ {\text{\;and\;}} V_s(x)\leq r\}$. The compactness follows from positive definiteness of $V_s$.  
	The set $\Omega$ is clearly forward invariant as $\dot{V_s}$ is nonpositive at any point within this set.
	Now, by LaSalle's invariance principle, solutions of the system initialized in $\Omega$ converge to the largest invariant set $\mathcal M$ in $\Omega$ where $\dot V_s= 0$. On this invariant set we have $\w=\1\w_s$.
	By using the second equality in \eqref{e:networkcshiftedS}, we obtain that
	\begin{equation}\label{e:inv-etaS}
	0=B\Gamma (\sin(\eta)-\sin(\eta_s))
	\end{equation}
	on the invariant set. Recall that $\eta, \eta_s \in \im B^T$, namely $\eta=B^T\theta$ and $\eta_s=B^T\theta_s$ for some vectors $\theta$ and 
	$\theta_s$.
	By multiplying \eqref{e:inv-etaS} from the left with $(\theta-\theta_s)^T$, we find that
	\[
	0=(\eta-\eta_s)^T\Gamma ({\sin}(\eta)-{\sin}(\eta_s)).
	\] 
	This results in $\eta=\eta_s$, as the compact level sets are constructed in a neighborhood of $\eta_s\in (-\frac{\pi}{2}, \frac{\pi}{2})^m$, where $\sin(\eta_k)$ is strictly monotone for each $k=1, 2, \ldots, m$. This completes the proof.
\end{proof}

\subsection{Secondary Control}
The primary controller stabilizes the system at the frequency $\w_s$, which in general is not equal to the nominal frequency $\w^*$. In this section, we aim to (optimally) regulate the frequency of the system \eqref{e:network} via the controller \eqref{e:u}, such that a unique equilibrium with $ \w_s=\w^*$ is achieved. Note that $\omega_s=\omega^\ast$ if and only if
\begin{align}\label{e:balance}
\1^T  P_m=\1^T P_\ell\;.
\end{align}
We associate a diagonal matrix $Q=\diag\{q_1,\ldots,q_n\}$ with the power generation costs, where $q_i \in \R_{>0}$ is the cost coefficient of the power generation of the $i$th inverter. Here we seek for an optimal resource allocation such that the control signal $ P_m=\col( P_{m_i})$ minimizes the quadratic cost function
\begin{align}\label{e:cost}
C(P_m)=\frac{1}{2}  P_m^TQ P_m \; \text{,}
\end{align}
subject to the power balance constraint given by \eqref{e:balance}.
Following the standard Lagrange multipliers method, the optimal control $P_m^\star$ that minimizes \eqref{e:cost} is computed as
\begin{align}\label{e:Opt}
P_m^\star&=\frac{Q^{-1}\1\1^T P_\ell}{\1^TQ^{-1}\1} \;.
\end{align}
An immediate consequence of the above is that the load is proportionally shared among the inverters, i.e, 
\begin{align}\label{e:pp}
(P^*_m)_iQ_i=(P^*_m)_jQ_j\; ,
\end{align}
for all $i, j\in \V$. To achieve the optimal cost, and inspired by \cite{Trip2016,Hierarchy2016, SimpsonPorco2013,Andreasson2014}, we propose the controller given by
\begin{equation}
\begin{aligned}
  \dot \xi&=-\la\xi-Q^{-1}[\w]^{-1}(\w-\1\w^*)\\
  P_m&=Q^{-1}\xi
  \;,
\end{aligned}\label{eq:controlleri}
\end{equation}
where $\la$ is the Laplacian matrix of an undirected connected communication graph. The term $-Q^{-1}[\w]^{-1}(\w-\1\w^*)$ regulates the frequency to the nominal frequency, while the consensus based algorithm $-\la\xi$ aims at steering the input to the optimal one given by \eqref{e:Opt}.
Having \eqref{e:network}-\eqref{e:u}, and \eqref{eq:controlleri}, the overall system reads as 
\begin{equation}\label{e:networkc}
\begin{aligned}
 \dot \eta&=B^T\w\\
J  \dot \w&=[\w]^{-1}(Q^{-1}\xi-P_\ell-B\Gamma\sin \eta)-D(\w-\1 \w^*)\\
  \dot \xi&=-\la\xi-Q^{-1}[\w]^{-1}(\w-\1\w^*)\;.
\end{aligned}
\end{equation}
Note that $\eta(0)=B^T\theta(0)$, and hence $\eta(t)\in\im B^T$ for all $t\geq0$.
Consequently, we can restrict the domain of solutions of \eqref{e:networkc} to $(\eta,\w,\xi)\in\mathcal X:=\im B^T\times\R^n\times \R^n$ which is clearly forward invariant. Next, we characterize the equilibrium of the above system. 
\begin{lemma}\label{l:exist}
 The point $(\bar\eta,\bar \w,\bar\xi)\in\mathcal X$ is an  equilibrium of \eqref{e:networkc} if and only if it satisfies
  \begin{equation}\label{eq:multiequilibria}
  \begin{aligned}
    &Q^{-1}\bar\xi-P_\ell-B\Gamma\sin \bar \eta=0\\
    &\bar \w=\1\w^*, \qquad \bar \xi=\frac{\1\1^TP_\ell}{\1^TQ^{-1}\1} \;.
  \end{aligned}
\end{equation}
\end{lemma}
\begin{proof}
  By the first equality in \eqref{e:networkc} it follows that $\w=\1\tilde{\w}$ for some $\tilde{\w}$. By premultiplying the third equality in \eqref{e:networkc} by $\1^T$ we obtain that $\1^TQ^{-1}[\w]^{-1}(\tilde{\w}-\w^*)=0$ which implies that $\tilde{\w}=\w^*$. In addition, $\bar \xi=\1 \tilde{\xi}$  for some $\tilde{\xi}\in\mathbb R$. Again, by premultiplying the second equation by $\1^T$ we obtain that $\1^TQ^{-1}\1\tilde{\xi}-\1^TP_\ell=0$ implying $\bar \xi=\frac{\1\1^TP_\ell}{\1^TQ^{-1}\1}$. 
\end{proof}

\medskip{}

Lemma \ref{l:exist} imposes the following assumption:
\begin{assumption}\label{assump:key}
For given $P_\ell$, there exists $\bar\eta\in  \im B^T \cap (-\frac{\pi}{2}, \frac{\pi}{2})^m$ such that
\[
\left(\frac{Q^{-1}\1\1^T}{\1^TQ^{-1}\1}-I_n\right)P_\ell-B\Gamma\sin \bar \eta=0\;.
\]
\end{assumption}
\vspace{0.3cm}
To prove frequency regulation, we exploit the energy function 
\begin{align}\label{e:V}
  W(x)=\frac12\w^TJ\w-{\w^*}^{-1}\1^T\Gamma \cos\eta+\frac12\xi^T\xi\;,
\end{align}
with $x=\col(\eta,\w,\xi)$. Note that the only difference with \eqref{e:VS} is the addition of the quadratic term associated with the states of the controller.
For the analysis, as before, we use the shifted version 
\begin{align}\label{e:W}
  W_s(x)=W(x)-(x-\bar x)^T\nabla W(\bar x)-W(\bar x)\;.
\end{align}
where $\bar x=(\bar\eta,\bar\w,\bar\xi)$. 
Noting that $\bar\eta\in   (-\frac{\pi}{2}, \frac{\pi}{2})^m$, it is easily verified that $W_s$ is positive definite around its local minimum $x=\bar x$.
\rem{By construction, $W_s$ is positive definite if the function $W$ is strictly convex \cite{Bregman1967}. Again it is easy to verify that $W_s$ takes it minimum at $x=\bar x$, providing that $\bar \eta \in  (-\frac{\pi}{2}, \frac{\pi}{2})^m$.  }
Now, we have the following result.
\begin{theorem}\label{thm:localasstab}
	Suppose that Assumption \ref{assump:key} holds. 
	Then there exists a neighborhood $\Omega$ of $(\bar\eta, \bar\omega, \bar\xi)$ such that any solution $(\eta, \omega, \xi)$ to \eqref{e:networkc} that starts in $\Omega$, asymptotically converges to the equilibrium point $(\bar\eta, \bar\omega, \bar\xi)$.
	Moreover, the vector $P_m$ converges to the optimal power injection $P_m^\star$ given by \eqref{e:Opt}.   
\end{theorem}
\begin{proof}
	First, observe that by substituting $P_{\ell}$ from \eqref{eq:multiequilibria} the system \eqref{e:networkc} can be written as 
	\begin{equation}\label{e:networkcshifted}
	\begin{aligned}
	\dot \eta=&B^T(\w-\1\w^*)\\
	J  \dot \w=&[\w]^{-1}(Q^{-1}(\xi-\bar \xi)
	\\&-B\Gamma(\sin \eta-\sin\bar\eta))-D(\w-\1 \w^*)\\
	\dot \xi=&-\la(\xi-\bar \xi)-Q^{-1}[\w]^{-1}(\w-\1\w^*)
	\end{aligned}
	\end{equation}
	Analogous to the proof of Theorem \ref{thm:localasstabS}, the time derivative of $W_s$ given by \eqref{e:V}-\eqref{e:W} along the solutions of \eqref{e:networkcshifted} is computed as
    \begin{align*}
	\dot{ W_s}=&-(\xi-\bar \xi)^T\la(\xi-\bar \xi)
	\\&-(\w-\1\w^*)^T\Big(D-{\w^*}^{-1}[\w]^{-1}[z(\eta)]\Big)(\w-\1\w^*)
	\end{align*}
	with
	\[
	z(\eta)=B\Gamma(\sin\eta-\sin\bar\eta)\;.
	\]
	Since $D>0$, $\w^*>0$, and $z(\bar\eta)=0$, there exists a neighborhood $\Omega^+$ around $(\bar \eta,\bar \w,\bar\xi)$ such that 
	\begin{align*}
	D-{\w^*}^{-1} [\w]^{-1}[z(\eta)]>0
	\end{align*}
	for all $(\eta,\w,\xi)\in \Omega^+$. 
	Take a (nontrivial) compact level set $\Omega$ of $W_s$ contained in this set, i.e. $\Omega\subset \Omega^+$. Again note that such $\Omega$ always exists for sufficiently small $r$, $\Omega=\{x \mid x\in \Omega^+ {\text{\;and\;}} W_s(x)\leq r\}$. 
	Noting that $\Omega$ is forward invariant, 
	by LaSalle's invariance principle, solutions of the system initialized in $\Omega$ converge to the largest invariant set $\mathcal M$ in $\Omega$ with $\dot W_s= 0$. On this invariant set we have $\w=\1\w^*$, and $\la\xi=0$ implying that $\xi=\bar \xi+\alpha \1$ for some $\alpha\in \R$.
	By premultiplying the second equality in \eqref{e:networkcshifted} with $\1^T$, on the invariant set we have
	$
	0=\1^TQ^{-1}(\bar \xi+ \alpha \1 - \bar \xi)
	$,  which yields $\alpha=0$ and thus $\xi=\bar \xi$. This means that $P_m=Q^{-1}\bar \xi$ on the invariant set, which coincides with the expression of optimal power injection $P_m^\star$ given by \eqref{e:Opt}, noting the last equality in \eqref{eq:multiequilibria}.
	Finally, by using an analogous argument to the proof of Theorem \ref{thm:localasstabS}, we conclude that $\eta=\overline \eta$ on the invariant set, which completes the proof.
\end{proof}

\begin{remark}\label{r:Network2Single}
	We can treat a single ICI modeled by \eqref{e:model}-\eqref{e:model0c} as the special case of the network modeled by \eqref{e:networkc} with $\la=0$, $n=1$, $Q=1$, and $\Gamma=0$. Hence the controller regulates the frequency to its nominal value also in the case of a single ICI connected to a constant load.
\end{remark}
\section{Numerical Example}
We illustrate the results by a numerical example
of a power network consisting of five ICIs. The interconnection topology (solid lines) and the communication graph (dashed lines) are shown in Figure \ref{f:numerical}. The reactance of the lines are depicted along the edges. The inverter setpoints and other network parameters are chosen as shown in Table \ref{t:param}.
\begin{figure}
	\[\begin{tikzpicture}[x=1.2cm, y=0.8cm,
	every edge/.style={sloped, draw, line width=1.2pt}]
	\vertex (v1) at (2,1)  {\tiny\textbf{ ICI\textsubscript{2}}};
	\vertex (v2) at (4,-1) {\tiny \textbf{ICI\textsubscript{4}}};
	\vertex (v3) at (2,-1) {\tiny \textbf{ICI\textsubscript{5}}};
	\vertex (v5) at (1,0) {\tiny \textbf{ICI\textsubscript{1}}};
	\vertex (v6) at (4,1) {\tiny \textbf{ICI\textsubscript{3}}};
	\tikzstyle{vertex}=[shading = ball, ball color = white!100!white, minimum size=12pt, draw, inner sep=0pt]
	\path
	(v1) edge node[anchor=south]{\weight{\rm 0.08}}(v5)
	(v1) edge node[anchor=south]{\weight{\rm 0.15}}(v6)
	(v2) edge node[anchor=north]{\weight{\rm 0.13}}(v3)
	(v2) edge node[anchor=south]{\weight{\rm 0.08}}(v6)
	(v3) edge node[anchor=north]{\weight{\rm 0.10}}(v5)
	(v1) edge[dashed] node[left]{}(v3)
	(v2) edge[dashed, bend right,in=210,out=330] node[left]{}(v3)
	(v3) edge[dashed, bend right,in=210,out=330] node[left]{}(v5)
	(v2) edge[dashed, bend right,in=210,out=330] node[right]{}(v6)
	;
	\end{tikzpicture}\]
	\caption{The solid lines denote the power lines in $\G$, and the dashed lines depict the communication links with the Laplacian $\la$. The values over the edges are the reactance of the lines.}\label{f:numerical}
\end{figure}
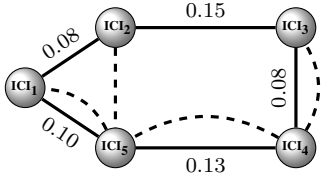
\begin{table}
	\centering
	\caption{Simulation Parameters}
	\label{t:param}
	\begin{tabular}{cccccc}
		\toprule
		& ICI\textsubscript{1}  & ICI\textsubscript{2}  & ICI\textsubscript{3} & ICI\textsubscript{4} & ICI\textsubscript{5} \\
		\midrule
		${C_{\dc}}_i(\si{\milli\farad})$ &1.0  & 1.2 & 1.1 &2.5&4.4\\[0.15cm]
		${G_{\dc}}_i(\si{\mho})$ & 0.10 & 0.09 & 0.12 & 0.12 & 0.18\\[0.15cm]
		$q_i(\$/\si{\kilo\watt\squared \hour})$ & 0.056&  0.028 &0.019&0.014&0.011\\[0.15cm]
		${P_\ell}_i(\si{\kilo\watt})$ & 10& 12.5&13.5 &16 &25\\[0.15cm]
		$|V_i|(\si{\volt})$ &300.7  &298.8 &299.7 & 301.0 &300.3\\[0.15cm]
		$v^*_{\dc}(\si{\kilo\volt})$ & 1.0 & 0.9 & 0.8& 1.2 & 1.5 \\[0.1cm]
		\bottomrule
	\end{tabular}
\end{table}
The system is initially at steady-state with the constant power loads $P_{\ell_i}$. At time $t=0$, loads $P_{\ell_1}$, $P_{\ell_3}$, and $P_{\ell_5}$ are increased by $10$ percent of their original values. The frequency evolution and
the active power injections are depicted in Figure \ref{fig:numerical}. It is observed that the system regulates the frequency
to its nominal value $\SI{50}{\hertz}$. Note that the frequencies at the various nodes are so similar to each other that no difference can be noticed in the plot. The system shows a safe maximum rate of change of frequency $\rm{ROCOF}_{\max}=\SI{0.3}{\hertz/\s}$ (ENTSOE standard threshold for the maximum $\rm ROCOF$ is $\SI{1}{\hertz/\s}$ \cite{ENTSOE}), which can be diminished further using larger or parallel capacitors. Finally, observe that the load is shared among the sources with the ratios of $\{q^{-1}_1,\cdots,q^{-1}_5\}$, which is in agreement with the proportional power sharing \eqref{e:pp}.
\begin{figure}
\centering
\includegraphics[width=8.7cm]{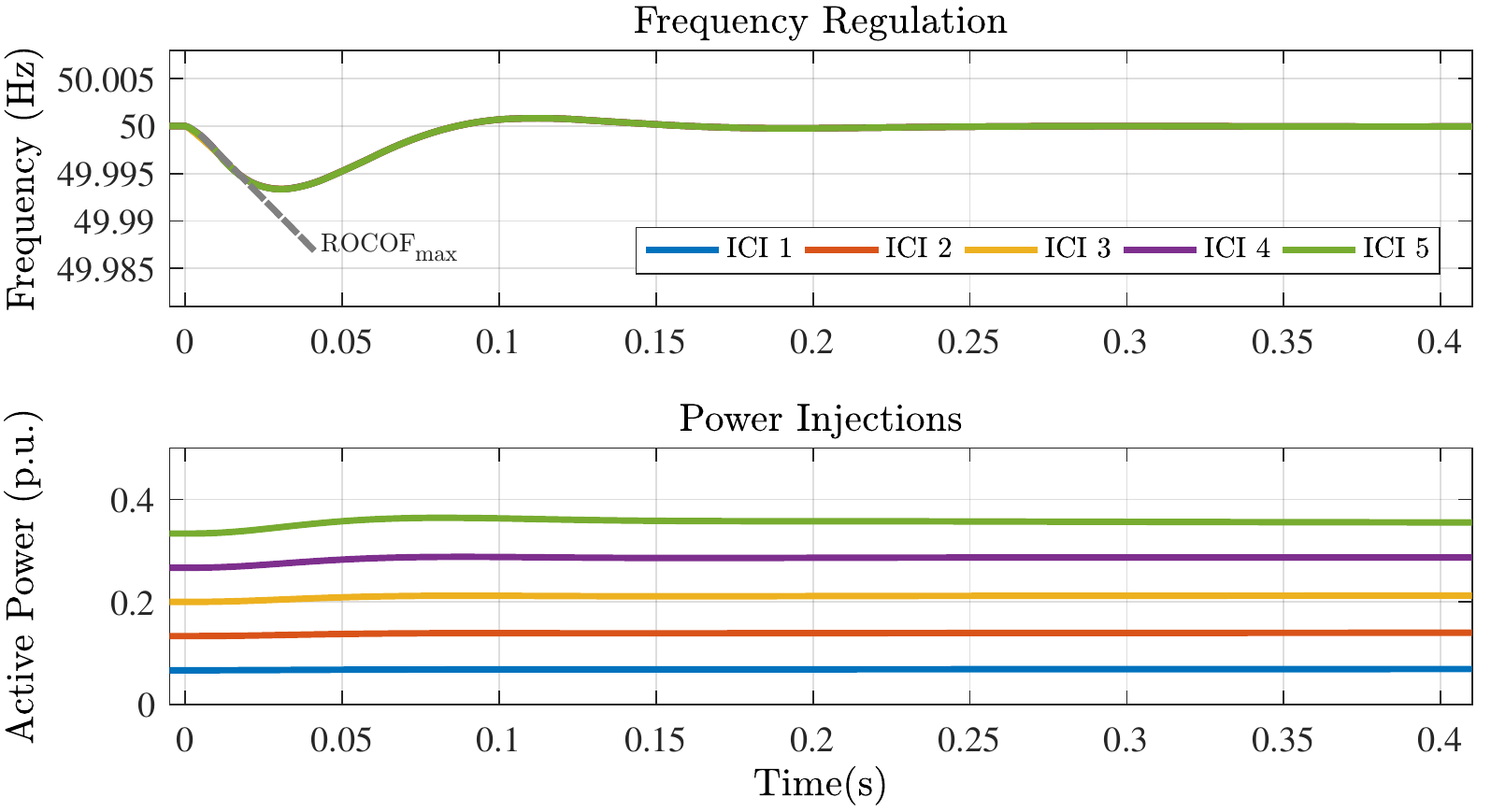}
\caption{Frequency regulation and optimal power injection after a step change in the local loads connected to the nodes $1$, $3$, and $5$.}
\label{fig:numerical}
\end{figure}

\section{Conclusion}
In this paper a network of inverters with a capacitor emulating inertia was investigated in two cases. First, the case of a single inverter connected to a load, and second, a network of inverters with local loads. A control method including primary and secondary controllers was proposed, and it was shown that the stability and frequency regulation are guaranteed under the proposed controllers. Future work includes the control of the reactive power, considering filter dynamics of the ICI \cite{Jouini2016}, time-domain analysis of the network, and extending the proposed results to structure-preserving and differential algebraic models \cite{BH, Varaiya1985,Nima2016_1,Nima2017_1}.
\bibliographystyle{IEEEtran}
\bibliography{ref,mybib}

\begin{thebibliography}{10}
\providecommand{\url}[1]{#1}
\csname url@samestyle\endcsname
\providecommand{\newblock}{\relax}
\providecommand{\bibinfo}[2]{#2}
\providecommand{\BIBentrySTDinterwordspacing}{\spaceskip=0pt\relax}
\providecommand{\BIBentryALTinterwordstretchfactor}{4}
\providecommand{\BIBentryALTinterwordspacing}{\spaceskip=\fontdimen2\font plus
\BIBentryALTinterwordstretchfactor\fontdimen3\font minus
  \fontdimen4\font\relax}
\providecommand{\BIBforeignlanguage}[2]{{%
\expandafter\ifx\csname l@#1\endcsname\relax
\typeout{** WARNING: IEEEtran.bst: No hyphenation pattern has been}%
\typeout{** loaded for the language `#1'. Using the pattern for}%
\typeout{** the default language instead.}%
\else
\language=\csname l@#1\endcsname
\fi
#2}}
\providecommand{\BIBdecl}{\relax}
\BIBdecl

\bibitem{Tielens2016}
P.~Tielens and D.~V. Hertem, ``The relevance of inertia in power systems,''
  \emph{Renewable and Sustainable Energy Reviews}, vol.~55, pp. 999 -- 1009,
  2016.

\bibitem{Ulbig2014}
A.~Ulbig, T.~S. Borsche, and G.~Andersson, ``Impact of low rotational inertia
  on power system stability and operation,'' \emph{IFAC Proceedings Volumes},
  vol.~47, no.~3, pp. 7290--7297, 2014.

\bibitem{Bevrani2014}
H.~Bevrani, T.~Ise, and Y.~Miura, ``Virtual synchronous generators: A survey
  and new perspectives,'' \emph{International Journal of Electrical Power \&
  Energy Systems}, vol.~54, pp. 244--254, 2014.

\bibitem{Dreidy2017}
M.~Dreidy, H.~Mokhlis, and S.~Mekhilef, ``Inertia response and frequency
  control techniques for renewable energy sources: A review,'' \emph{Renewable
  and Sustainable Energy Reviews}, vol.~69, pp. 144 -- 155, 2017.

\bibitem{vi-4}
J.~Ekanayake and N.~Jenkins, ``Comparison of the response of doubly fed and
  fixed-speed induction generator wind turbines to changes in network
  frequency,'' \emph{IEEE Transactions on Energy Conversion}, vol.~19, no.~4,
  pp. 800--802, 2004.

\bibitem{vi-3}
J.~Morren, J.~Pierik, and S.~W. de~Haan, ``Inertial response of variable speed
  wind turbines,'' \emph{Electric Power Systems Research}, vol.~76, no.~11, pp.
  980 -- 987, 2006.

\bibitem{vi-2}
H.~P. Beck and R.~Hesse, ``Virtual synchronous machine,'' in \emph{9th
  International Conference on Electrical Power Quality and Utilisation}, 2007,
  pp. 1--6.

\bibitem{vi-1}
M.~P.~N. van Wesenbeeck, S.~W.~H. de~Haan, P.~Varela, and K.~Visscher, ``Grid
  tied converter with virtual kinetic storage,'' in \emph{IEEE Bucharest
  PowerTech}, 2009, pp. 1--7.

\bibitem{vi0}
T.~V. Van, K.~Visscher, J.~Diaz, V.~Karapanos, A.~Woyte, M.~Albu, J.~Bozelie,
  T.~Loix, and D.~Federenciuc, ``Virtual synchronous generator: An element of
  future grids,'' in \emph{IEEE PES Innovative Smart Grid Technologies
  Conference Europe (ISGT Europe)}, 2010, pp. 1--7.

\bibitem{vi1}
Q.~C. Zhong and G.~Weiss, ``Synchronverters: Inverters that mimic synchronous
  generators,'' \emph{IEEE Transactions on Industrial Electronics}, vol.~58,
  no.~4, pp. 1259--1267, 2011.

\bibitem{vi2}
N.~Soni, S.~Doolla, and M.~C. Chandorkar, ``Improvement of transient response
  in microgrids using virtual inertia,'' \emph{IEEE Transactions on Power
  Delivery}, vol.~28, no.~3, pp. 1830--1838, 2013.

\bibitem{vi3}
T.~Shintai, Y.~Miura, and T.~Ise, ``Oscillation damping of a distributed
  generator using a virtual synchronous generator,'' \emph{IEEE Transactions on
  Power Delivery}, vol.~29, no.~2, pp. 668--676, 2014.

\bibitem{vi4}
J.~Alipoor, Y.~Miura, and T.~Ise, ``Power system stabilization using virtual
  synchronous generator with alternating moment of inertia,'' \emph{IEEE
  Journal of Emerging and Selected Topics in Power Electronics}, vol.~3, no.~2,
  pp. 451--458, 2015.

\bibitem{Vandoorn2011}
T.~L. Vandoorn, B.~Meersman, L.~Degroote, B.~Renders, and L.~Vandevelde, ``A
  control strategy for islanded microgrids with {DC}-link voltage control,''
  \emph{IEEE Transactions on Power Delivery}, vol.~26, no.~2, pp. 703--713,
  2011.

\bibitem{Vandoorn2012}
T.~L. Vandoorn, B.~Meersman, J.~D. M.~D. Kooning, and L.~Vandevelde, ``Analogy
  between conventional grid control and islanded microgrid control based on a
  global {DC}-link voltage droop,'' \emph{IEEE Transactions on Power Delivery},
  vol.~27, no.~3, pp. 1405--1414, 2012.

\bibitem{Arani2013}
M.~F.~M. Arani and E.~F. El-Saadany, ``Implementing virtual inertia in
  {DFIG}-based wind power generation,'' \emph{IEEE Transactions on Power
  Systems}, vol.~28, no.~2, pp. 1373--1384, 2013.

\bibitem{Jouini2016}
T.~Jouini, C.~Arghir, and F.~D{\"o}rfler, ``Grid-friendly matching of
  synchronous machines by tapping into the {DC} storage,''
  \emph{IFAC-PapersOnLine}, vol.~49, no.~22, pp. 192--197, 2016.

\bibitem{Zhou2009}
J.~Zhou and Y.~Ohsawa, ``Improved swing equation and its properties in
  synchronous generators,'' \emph{Circuits and Systems I: Regular Papers, IEEE
  Transactions on}, vol.~56, no.~1, pp. 200--209, 2009.

\bibitem{Pooya2016}
P.~Monshizadeh, C.~De~Persis, N.~Monshizadeh, and A.~J. van~der Schaft,
  ``Nonlinear analysis of an improved swing equation,'' in \emph{IEEE 55th
  Conference on Decision and Control (CDC)}, 2016, pp. 4116--4121.

\bibitem{Schiffer2014}
J.~Schiffer, R.~Ortega, A.~Astolfi, J.~Raisch, and T.~Sezi, ``Conditions for
  stability of droop-controlled inverter-based microgrids,'' \emph{Automatica},
  vol.~50, no.~10, pp. 2457 -- 2469, 2014.

\bibitem{Pooya2017}
P.~Monshizadeh, N.~Monshizadeh, C.~De~Persis, and A.~van~der Schaft, ``Output
  impedance diffusion into lossy power lines,'' \emph{arXiv preprint
  arXiv:1702.01488}, 2017.

\bibitem{Hierarchy2016}
F.~D\"{o}rfler, J.~W. Simpson-Porco, and F.~Bullo, ``Breaking the hierarchy:
  Distributed control and economic optimality in microgrids,'' \emph{IEEE
  Transactions on Control of Network Systems}, vol.~3, no.~3, pp. 241--253,
  2016.

\bibitem{Bregman1967}
L.~Bregman, ``The relaxation method of finding the common point of convex sets
  and its application to the solution of problems in convex programming,''
  \emph{USSR Computational Mathematics and Mathematical Physics}, vol.~7,
  no.~3, pp. 200 -- 217, 1967.

\bibitem{Bayu2007}
B.~Jayawardhana, R.~Ortega, E.~Garcia-Canseco, and F.~Castanos, ``Passivity of
  nonlinear incremental systems: Application to {PI} stabilization of nonlinear
  {RLC} circuits,'' \emph{Systems \& control letters}, vol.~56, no.~9, pp.
  618--622, 2007.

\bibitem{Nima2015}
C.~De~Persis and N.~Monshizadeh, ``Bregman storage functions for microgrid
  control,'' \emph{IEEE Transactions on Automatic Control, provisionally
  accepted}, 2015.

\bibitem{Trip2016}
S.~Trip, M.~B{\"u}rger, and C.~De~Persis, ``An internal model approach to
  (optimal) frequency regulation in power grids with time-varying voltages,''
  \emph{Automatica}, vol.~64, pp. 240--253, 2016.

\bibitem{SimpsonPorco2013}
J.~W. Simpson-Porco, F.~D\"{o}rfler, and F.~Bullo, ``Synchronization and power
  sharing for droop-controlled inverters in islanded microgrids,''
  \emph{Automatica}, vol.~49, no.~9, pp. 2603 -- 2611, 2013.

\bibitem{Andreasson2014}
M.~Andreasson, D.~V. Dimarogonas, H.~Sandberg, and K.~H. Johansson,
  ``Distributed {PI}-control with applications to power systems frequency
  control,'' in \emph{IEEE American Control Conference}, 2014, pp. 3183--3188.

\bibitem{ENTSOE}
{European Network of Transmission System Operators for Electricity (ENTSOE)},
  ``{Frequency Stability Evaluation Criteria for the Synchronous Zone of
  Continental {E}urope - Requirements and impacting factors},''
  \emph{{Distribution System Analysis Subcommittee}}, 2016.

\bibitem{BH}
A.~R. Bergen and D.~J. Hill, ``A structure preserving model for power system
  stability analysis,'' \emph{IEEE Transactions on Power Apparatus and
  Systems}, vol. PAS-100, no.~1, pp. 25--35, 1981.

\bibitem{Varaiya1985}
N.~Tsolas, A.~Arapostathis, and P.~Varaiya, ``A structure preserving energy
  function for power system transient stability analysis,'' \emph{IEEE
  Transactions on Circuits and Systems}, vol.~32, no.~10, pp. 1041--1049, 1985.

\bibitem{Nima2016_1}
C.~De~Persis, N.~Monshizadeh, J.~Schiffer, and F.~D\"{o}rfler, ``A {Lyapunov}
  approach to control of microgrids with a network-preserved
  differential-algebraic model,'' in \emph{IEEE 55th Conference on Decision and
  Control (CDC)}, 2016, pp. 2595--2600.

\bibitem{Nima2017_1}
N.~Monshizadeh and C.~D. Persis, ``Agreeing in networks: Unmatched
  disturbances, algebraic constraints and optimality,'' \emph{Automatica},
  vol.~75, pp. 63 -- 74, 2017.

\end{thebibliography}
\IEEEpeerreviewmaketitle
\end{document}